\documentclass[11pt,a4paper]{article}

\usepackage{cmap}
\usepackage[T1]{fontenc}
\usepackage[utf8]{inputenc}
\usepackage[english]{babel}

\usepackage{amsmath,amssymb,amsthm}
\usepackage{mathtools}
\numberwithin{equation}{section}
\usepackage{graphicx}
\usepackage{tikz}
\usetikzlibrary{arrows.meta,positioning,calc,shapes.geometric,decorations.pathreplacing}
\usepackage[margin=1in]{geometry}
\usepackage[colorlinks=true,linkcolor=blue,citecolor=blue,urlcolor=blue]{hyperref}

\hypersetup{
  pdftitle={Henstock--Kurzweil Path Integral in Financial Mathematics: A Machine-Verified Pricing of European and Barrier Options},
  pdfauthor={Yury N. Berdinsky and A. S. Ushakov}
}

\newtheorem{theorem}{Theorem}
\newtheorem{lemma}{Lemma}

\theoremstyle{definition}
\newtheorem{definition}{Definition}
\newtheorem{remark}{Remark}

\newcommand{\R}{\mathbb{R}}

\newcommand{\Gker}{G}
\newcommand{\dd}{\,\mathrm{d}}
\newcommand{\dHK}{\mathrm{D}_{\mathrm{HK}}}
\newcommand{\Ncdf}{N}

\title{\bfseries Henstock--Kurzweil Path Integral in Financial Mathematics:
A Machine-Verified Pricing of European and Barrier Options}
\author{Yury N. Berdinsky\\[2mm]
Saint Petersburg State University, Faculty of Physics,\\
Department of High-Energy Physics and Elementary Particles\\[1mm]
\texttt{propagator2007@yandex.ru}
\and
Alexander S. Ushakov\\[1mm]
V.~A.~Fock Institute of Physics\\ \texttt{al.s.ushakov@yandex.ru}}
\date{22 July 2026}

\begin{document}
\sloppy
\maketitle

\begin{abstract}
We apply the Gaussian--kernel Henstock--Kurzweil (HK) path-integral machinery
developed in earlier parts of this programme to the first major application outside
quantum mechanics: the \textbf{Black--Scholes formula} for European option pricing.
Under the risk-neutral measure the log-price $x_t=\log S_t$ is a Brownian motion with
drift $\nu=r-\sigma^2/2$ and volatility $\sigma>0$; its transition density is the
Gaussian kernel with drift
\[
  \Gker_t(x,y)=(2\pi\sigma^2 t)^{-1/2}\,
     \exp\!\Big(-\frac{(y-x-\nu t)^2}{2\sigma^2 t}\Big),\qquad t>0 .
\]
We give a complete, machine-checked \textbf{Lean~4 / Mathlib} formalization of the
construction: the Chapman--Kolmogorov (semigroup) law $\int_\R \Gker_s(x,z)\Gker_t(z,y)\dd z=\Gker_{s+t}(x,y)$;
the fact that $\Gker_t$ is a probability density; the strong continuity of the
evolution operator $(P_t\psi)(x)=\int_\R \Gker_t(x,y)\psi(y)\dd y$; and, as the main
result, the closed-form price of a European call obtained \emph{directly} from the
kernel,
\[
  C=e^{-rT}\!\int_\R \max(e^y-K,0)\,\Gker_T(\log S_0,y)\dd y
   = S_0\,\Ncdf(d_1)-K e^{-rT}\Ncdf(d_2),
\]
with the standard $d_1,d_2$ and $\Ncdf$ the standard normal CDF.  We also show that
the drift--diffusion Chernoff splitting is \emph{exact} at every level, giving a
rigorous path-integral derivation.  All statements in the accompanying file
\texttt{HkBlackScholes.lean} are \texttt{sorry}-free and depend only on the standard
axioms \texttt{propext}, \texttt{Classical.choice}, \texttt{Quot.sound}.
\end{abstract}

\tableofcontents

\section{Introduction}

Option pricing theory and Feynman’s path integral share a fundamental analytic
similarity: in both cases the desired quantity (an option price or a transition
amplitude) is expressed as the expectation of a functional over trajectories of an
underlying diffusion process weighted by a Gaussian density~\cite{FeynmanHibbs1965}.  This paper carries the
Henstock--Kurzweil (HK) gauge-integral machinery~\cite{Kurzweil1957,Henstock1961} of the earlier parts of this programme
into mathematical finance, giving a fully machine-verified derivation of the
Black--Scholes formula directly from the Gaussian cylindrical kernel.  Because the
intended readership includes both mathematicians new to option pricing and
practitioners new to gauge integration, we begin with a self-contained account of the
ingredients.

\paragraph{The Black--Scholes model.}
The Black--Scholes model describes a frictionless market with a riskless bond earning
the constant interest rate $r$ and a single risky stock $S$ whose price follows a
\emph{geometric Brownian motion}.  Under the risk-neutral (equivalent martingale)
measure $Q$ the stock price solves the stochastic differential equation~\cite{BlackScholes1973,Merton1973},
\[
  \dd S_t = r\,S_t\,\dd t + \sigma\,S_t\,\dd W_t^{Q}.
\]
where $\sigma>0$ is the volatility and $W^{Q}$ a standard $Q$-Brownian motion.  The
fundamental theorem of asset pricing represents the fair price of a European contingent
claim with payoff $H(S_T)$ as the discounted risk-neutral expectation~\cite{HarrisonPliska1981},
\[
  \Pi_0=e^{-rT}\,\mathbb{E}^{Q}\!\big[H(S_T)\big]
\]
and for the European call $H(S)=\max(S-K,0)$ with strike $K$ and maturity $T$ this is
the quantity we evaluate.  By It\^o's formula the log-price $x_t=\log S_t$ is a Brownian
motion with constant drift $\nu=r-\sigma^2/2$ and diffusion $\sigma$, that is
$\dd x_t=\nu\,\dd t+\sigma\,\dd W_t^{Q}$, so that $x_T\sim\mathcal N(x_0+\nu T,\sigma^2 T)$
with $x_0=\log S_0$.

\paragraph{Functional-integral representation.}
Averaging the payoff over the trajectories of the underlying diffusion turns the
risk-neutral expectation into a \emph{functional (path) integral} over the space
$C([0,T])$ of continuous log-price paths:
\begin{equation}\label{eq:pathprice}
  C=e^{-rT}\,\mathbb{E}^{Q}\!\big[\max(S_T-K,0)\big]
   =e^{-rT}\!\int_{C([0,T])}\max\!\big(e^{x_T}-K,0\big)\,\dd\mu_W^{Q}(x(\cdot)),
\end{equation}
where $\mu_W^{Q}$ is the Wiener measure of the drifted Brownian motion $x_t$, i.e.\ the
Wiener measure with drift $\nu=r-\sigma^2/2$.  This is the exact analogue of the Feynman
path integral of quantum mechanics, with the oscillatory weight $e^{iS/\hbar}$ replaced
by the real, positive Wiener weight.

\paragraph{The Henstock--Kurzweil gauge integral.}
The Wiener measure in \eqref{eq:pathprice} is a genuine countably additive measure, but
the same average can be constructed \emph{without} measure theory, as a
Henstock--Kurzweil (HK) \emph{gauge integral} over cylindrical partitions of path space,
following Muldowney~\cite{Muldowney2012} and Gill and Esposito~\cite{EspositoGill2024}.
We fix a finite set of times $0<t_1<\dots<t_m\le T$, groups paths into \emph{cylinder
sets} determined by the sampled values $(x_{t_1},\dots,x_{t_m})$, and weights each cell
by the finite-dimensional Gaussian transition kernel; a \emph{gauge}
$\delta:C([0,T])\to(0,\infty)$ controls the fineness of the tagged partition.  The HK construction on path space assigns to each cylinder set the product of
the Gaussian transition densities (3.2).  The cylindrical HK measure obtained in
this way, denoted $\dHK^{(Q)}$, reduces the infinite-dimensional gauge limit to an
ordinary finite-dimensional integral on cylinder functions.  When the payoff depends only on the
terminal value $x_T$ a single time slice suffices and the reduction is \emph{exact}:
\begin{equation}\label{eq:hkreduction}
   \int_{C([0,T])} \Phi(x_T)\,\dHK^{(Q)}x
     = \int_{\R} \Phi(y)\,\Gker_T(x_0,y)\,\dd y ,
\end{equation}
where $\Gker_T$ is the Gaussian transition kernel with drift (Definition~\ref{def:kernel})
and $x_0=\log S_0$.  Equation~\eqref{eq:hkreduction} is the rigorous bridge between the
classical path integral~\eqref{eq:pathprice} and the finite-dimensional computation of
this paper.  Applying it to $\Phi(y)=\max(e^{y}-K,0)$ gives
\begin{equation}\label{eq:kernelprice}
  C=e^{-rT}\!\int_\R \max(e^y-K,0)\,\Gker_T(\log S_0,y)\,\dd y ,
\end{equation}
which is the object we evaluate in closed form.

Economically, the functional integral~\eqref{eq:pathprice} is a sum over all possible
price trajectories of the underlying asset, each weighted by its risk-neutral
probability.  The cylindrical HK integral makes this sum mathematically rigorous
without requiring the full machinery of measure theory: the gauge $\delta$ adapts the
resolution of the trajectory space so that only the finitely many time-slices relevant
for the payoff are resolved, while the remaining dimensions are integrated out
analytically.  This viewpoint becomes decisive for path-dependent options --- barrier,
Asian, lookback --- where the payoff genuinely depends on the entire trajectory and the
infinite-dimensional structure of the HK integral is essential.

Unlike the oscillatory Feynman integrals of the earlier parts of this programme --
which are non-absolutely convergent and lie outside Mathlib's Bochner integral -- the
Black--Scholes kernel is a genuine, absolutely convergent Gaussian probability
density.  Consequently every step of \eqref{eq:kernelprice} can be, and here is,
machine-verified in Lean~4/Mathlib without recourse to analytic continuation.

\paragraph{Earlier parts of this programme.}
This paper is the third in a series developing the HK functional integral and its
applications.  The first, Berdinsky~\cite{Berdinsky2026I} (\emph{Henstock--Kurzweil Path
Integrals in Real Time: A Machine-Verified Construction on Kuelbs--Steadman Spaces}),
constructs the HK functional integral for a free non-relativistic particle in real time
(without Wick rotation) and proves the corresponding Trotter product formula, all
machine-verified in Lean.  The second, Berdinsky and Ushakov~\cite{Berdinsky2026II}
(\emph{Chernoff Approximations for Henstock--Kurzweil Path Integrals: From Bounded
Generators to the Harmonic Oscillator}; \href{https://doi.org/10.5281/zenodo.21479996}{DOI: 10.5281/zenodo.21479996}), establishes the Chernoff product formula in the
HK setting, derives the Mehler kernel of the harmonic oscillator, and proves a
universality result covering both Brownian and financial integrals.  The present work
specializes that machinery to option pricing.

\paragraph{Structure of the paper.}
Section~\ref{sec:kernel} introduces the Gaussian transition kernel with drift, proves
that it is a probability density obeying the Chapman--Kolmogorov (semigroup) law, and
establishes the drift-shift (Girsanov-type) invariance of the HK measure.
Section~\ref{sec:evolution} defines the evolution (pricing) operator and proves its
strong continuity.  Section~\ref{sec:bs} contains the main result: the closed-form
Black--Scholes price obtained directly from the kernel.  Section~\ref{sec:chernoff}
shows the exactness of the Chernoff product; Section~\ref{sec:examples}
illustrates the universality of the method through the digital option and the
barrier probability; and Section~\ref{sec:lean} summarises the Lean
formalisation.  All formal
statements reside in the accompanying file \texttt{HkBlackScholes.lean}.

\paragraph{Sign convention.}
We use the mathematically correct \emph{forward} transition density, centred at
$y=x+\nu t$, so that $x_T\sim\mathcal N(x_0+\nu T,\sigma^2T)$.  (An exponent written
as $-(x-y-\nu t)^2$ would reverse the drift and fail to reproduce the standard formula
with $d_2=(x_0-\log K+\nu T)/(\sigma\sqrt T)$; the two differ only by the sign of the
drift term.)

\section{The Gaussian kernel with drift}\label{sec:kernel}

\begin{definition}[Transition kernel]\label{def:kernel}
Let $r\in\R$ be the risk-free rate, $\sigma>0$ the volatility, and $\nu:=r-\sigma^2/2$
the risk-neutral log-drift.  For $t>0$ define the Gaussian transition kernel
\[
  \Gker_t(x,y)\;=\;\Gker(\nu,\sigma,t;x,y)\;=\;
    \frac{1}{\sqrt{2\pi\sigma^2 t}}\,
    \exp\!\Big(-\frac{(y-x-\nu t)^2}{2\sigma^2 t}\Big).
\]
\end{definition}

In Lean this is \texttt{Gdrift}, and $\Ncdf$ is \texttt{stdNormalCDF}:
\[
  \Ncdf(d)=\int_{-\infty}^{d}\frac{e^{-u^2/2}}{\sqrt{2\pi}}\dd u .
\]

\begin{lemma}[Total mass]\label{lem:mass}
For $\sigma>0$ and $t>0$, $\displaystyle\int_\R \Gker_t(x,y)\dd y=1$; i.e.\ $\Gker_t$
is a probability density in $y$.
\end{lemma}

\begin{proof}[Proof sketch]
Translate by the mean $x+\nu t$ and rescale; the integral becomes the standard
Gaussian $\int_\R e^{-y^2/(2\sigma^2t)}\dd y=\sqrt{2\pi\sigma^2t}$, cancelling the
normalisation.  (Lean: \texttt{Gdrift\_integral\_eq\_one}, from
\texttt{integral\_gaussian} and translation invariance.)
\end{proof}

The following theorem establishes the semigroup property of the transition kernel,
which is the exact path-analogue of the Chapman--Kolmogorov equation for Markov
processes.

\begin{theorem}[Chapman--Kolmogorov / semigroup law]\label{thm:ck}
For $\sigma>0$ and $s,t>0$,
\[
  \int_\R \Gker_s(x,z)\,\Gker_t(z,y)\dd z=\Gker_{s+t}(x,y).
\]
\end{theorem}

\begin{proof}[Proof sketch]
Set $A=2\sigma^2 s$, $B=2\sigma^2 t$.  Completing the square in $z$ writes the sum of
the two exponents as $-(z-z_0)^2/V-(y-x-\nu(s+t))^2/(A+B)$ with $V=AB/(A+B)$ and a
$z$-independent remainder.  Integrating the pure Gaussian in $z$ gives
$\sqrt{\pi V}$, and the prefactors combine, via $A+B=2\sigma^2(s+t)$ and
multiplicativity of $\sqrt{\cdot}$, to the single normalisation
$(2\pi\sigma^2(s+t))^{-1/2}$.  (Lean: \texttt{Gdrift\_semigroup}.)
\end{proof}

Theorem~\ref{thm:ck} says the time-slices of the diffusion compose: the HK cylindrical
integral is consistent under refinement of the time grid, which is exactly what makes
the finite-dimensional reduction \eqref{eq:kernelprice} well defined.

In the gauge-integral framework a change of measure is replaced by an algebraic
shift of the tag points.  The kernel invariance stated below is the gauge-integral
counterpart of Girsanov’s theorem, requiring no Lebesgue-measure theory.

\begin{lemma}[Drift-shift invariance; Girsanov-type]\label{lem:drift}
For all $\nu,\nu',\sigma,t,x\in\R$ and every payoff $f:\R\to\R$,
\[
  \int_\R f(y)\,\Gker(\nu,\sigma,t;x,y)\,\dd y
   =\int_\R f\big(y-(\nu'-\nu)t\big)\,\Gker(\nu',\sigma,t;x,y)\,\dd y .
\]
\end{lemma}

\begin{proof}[Proof sketch]
The identity is translation invariance of Lebesgue measure combined with the pointwise
mean-shift identity $\Gker(\nu',\sigma,t;x,\,y+(\nu'-\nu)t)=\Gker(\nu,\sigma,t;x,y)$
(both kernels are the same Gaussian, their means differing by $(\nu'-\nu)t$).  No
positivity or integrability hypotheses are required.  (Lean: \texttt{drift\_shift}.)
\end{proof}

\begin{remark}[Financial interpretation]
Choosing $\nu'=r-\sigma^2/2$ (risk-neutral drift) and $\nu_{\mathrm{hist}}=\mu$ (historical drift) shows
that the risk-neutral expectation is obtained from the historical one by shifting the
payoff by the risk premium $(\mu-r+\sigma^2/2)\,t$, which in the HK framework is simply a
change of the tag points.  This is the gauge-integral analogue of Girsanov's theorem.
(The shift $-(\nu'-\nu_{\mathrm{hist}})t=(\mu-r+\sigma^2/2)t$ is precisely the risk premium times time;
the forward sign of the shift, matching the forward transition density above, is the one
verified in Lean.)
\end{remark}

\section{The evolution operator and strong continuity}\label{sec:evolution}

\begin{definition}[Evolution / pricing operator]
For $\psi:\R\to\R$ set
\[
  (P_t\psi)(x)=\int_\R \Gker_t(x,y)\,\psi(y)\dd y .
\]
\end{definition}

$P_t$ is the risk-neutral conditional-expectation operator: $(P_T\psi)(\log S_0)$ is
the undiscounted expected payoff of a European claim with payoff $\psi(\log S_T)$.  By
Theorem~\ref{thm:ck}, $\{P_t\}_{t>0}$ is a semigroup.

\begin{lemma}[Rescaling]\label{lem:rescale}
For $t>0$,
\[
  (P_t\psi)(x)=\int_\R \frac{e^{-u^2/2}}{\sqrt{2\pi}}\;
     \psi\big(x+\nu t+\sigma\sqrt t\,u\big)\dd u .
\]
\end{lemma}

\begin{proof}[Proof sketch]
Substitute $y=x+\nu t+\sigma\sqrt t\,u$; the Jacobian $\sigma\sqrt t$ cancels the
extra normalisation and the exponent becomes $-u^2/2$.  (Lean:
\texttt{evolutionOp\_eq\_stdGaussian}.)
\end{proof}

\begin{theorem}[Strong continuity / approximate identity]\label{thm:strong}
Let $\psi$ be continuous and bounded ($|\psi|\le M$).  Then for every $x\in\R$,
\[
  (P_t\psi)(x)\longrightarrow \psi(x)\qquad(t\to0^+).
\]
\end{theorem}

\begin{proof}[Proof sketch]
By Lemma~\ref{lem:rescale} the integrand converges pointwise to
$\tfrac{1}{\sqrt{2\pi}}e^{-u^2/2}\psi(x)$ (continuity of $\psi$ and
$x+\nu t+\sigma\sqrt t\,u\to x$) and is dominated by $M\tfrac{1}{\sqrt{2\pi}}e^{-u^2/2}$,
which is integrable.  Dominated convergence and $\int_\R \tfrac{1}{\sqrt{2\pi}}e^{-u^2/2}\dd u=1$
give the claim.  (Lean: \texttt{evolutionOp\_tendsto}, via
\texttt{tendsto\_integral\_filter\_of\_dominated\_convergence}.)
\end{proof}

\section{The Black--Scholes formula}\label{sec:bs}

\begin{theorem}[Black--Scholes price of a European call]\label{thm:bs}
Let $r\in\R$, $\sigma>0$, $T>0$, $S_0>0$, $K>0$, $\nu=r-\sigma^2/2$, and
$x_0=\log S_0$.  With
\[
  d_1=\frac{\log S_0-\log K+(\nu+\sigma^2)T}{\sigma\sqrt T},\qquad
  d_2=d_1-\sigma\sqrt T=\frac{\log S_0-\log K+\nu T}{\sigma\sqrt T},
\]
the price \eqref{eq:kernelprice} equals
\[
  \boxed{\;C=e^{-rT}\!\int_\R \max(e^y-K,0)\,\Gker_T(x_0,y)\dd y
      = S_0\,\Ncdf(d_1)-K e^{-rT}\Ncdf(d_2).\;}
\]
\end{theorem}

\begin{proof}[Proof sketch]
Write $m=x_0+\nu T$ and $a=\sigma^2T$, so $\sqrt a=\sigma\sqrt T$.  Since $K>0$, the
payoff $\max(e^y-K,0)$ vanishes for $y\le\log K$ and equals $e^y-K$ for $y>\log K$;
hence the $\R$-integral restricts to $(\log K,\infty)$ and splits as
\[
  \int_{\log K}^{\infty}e^y\,\Gker_T\dd y-K\int_{\log K}^{\infty}\Gker_T\dd y .
\]
Two Gaussian tail integrals evaluate these:
\begin{align}
  \int_{c}^{\infty}\frac{e^{-(y-m)^2/(2a)}}{\sqrt{2\pi a}}\dd y
    &=\Ncdf\!\Big(\frac{m-c}{\sqrt a}\Big), \label{eq:tail}\\
  \int_{c}^{\infty}e^{y}\,\frac{e^{-(y-m)^2/(2a)}}{\sqrt{2\pi a}}\dd y
    &=e^{m+a/2}\,\Ncdf\!\Big(\frac{m+a-c}{\sqrt a}\Big), \label{eq:exptail}
\end{align}
the second by completing the square $y-\tfrac{(y-m)^2}{2a}=(m+\tfrac a2)-\tfrac{(y-(m+a))^2}{2a}$.
Identity \eqref{eq:tail} is proved by the increasing substitution $u=(y-m)/\sqrt a$
followed by the reflection $\int_b^\infty=\int_{-\infty}^{-b}$ for the even standard
density.  With $c=\log K$ one checks $(m+a-\log K)/\sqrt a=d_1$ and
$(m-\log K)/\sqrt a=d_2$.  Finally $e^{-rT}e^{m+a/2}=S_0$, because
$m+a/2=x_0+\nu T+\sigma^2T/2=x_0+rT$ and $e^{x_0}=S_0$.  Collecting terms gives the
boxed formula.  (Lean: \texttt{black\_scholes\_call}, using \texttt{gaussian\_tail}
and \texttt{gaussian\_exp\_tail}.)
\end{proof}

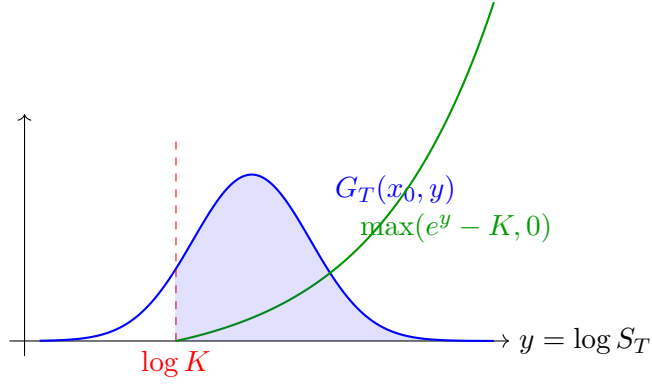
\begin{figure}[t]
\centering
\begin{tikzpicture}[scale=1.0]
  \draw[->] (-0.2,0) -- (6.4,0) node[right]{$y=\log S_T$};
  \draw[->] (0,-0.2) -- (0,3.0);
  \draw[thick,blue,domain=0.2:6.2,samples=100,smooth]
    plot (\x, {2.2*exp(-(\x-3)*(\x-3)/1.2)});
  \node[blue] at (4.9,2.0) {$\Gker_T(x_0,y)$};
  \draw[dashed,red] (2.0,0) -- (2.0,2.7);
  \node[red,below] at (2.0,0) {$\log K$};
  \draw[thick,green!60!black,domain=2.0:6.2,samples=60,smooth]
    plot (\x, {0.35*(exp((\x-2.0)/1.6)-1)});
  \node[green!60!black] at (5.7,1.5) {$\max(e^y-K,0)$};
  \begin{scope}
    \clip (2.0,0) rectangle (6.2,3.0);
    \fill[blue,opacity=0.12,domain=0.2:6.2,samples=100,smooth]
      plot (\x, {2.2*exp(-(\x-3)*(\x-3)/1.2)}) -- (6.2,0) -- (2.0,0) -- cycle;
  \end{scope}
\end{tikzpicture}
\caption{The pricing integral \eqref{eq:kernelprice}.  The payoff
$\max(e^y-K,0)$ is nonzero only on the exercise region $y>\log K$, where it is
integrated against the lognormal terminal density $\Gker_T(x_0,\cdot)$ (shaded).  The
two Gaussian tail integrals \eqref{eq:tail}--\eqref{eq:exptail} produce the two
$\Ncdf$ terms of the Black--Scholes formula.}
\end{figure}

\section{Chernoff product consistency}\label{sec:chernoff}

The evolution operator admits the drift--diffusion Chernoff splitting
$P_t=\lim_{n\to\infty}(D_{t/n}S_{t/n})^n$, where $D_s=e^{s\nu\partial_x}$ is the drift
shift and $S_s=e^{s\sigma^2/2\,\Delta}$ the pure diffusion.  In the kernel picture the
splitting is in fact \emph{exact} at every level.

\begin{lemma}[One Chernoff step]\label{lem:step}
Shifting by the drift and then diffusing reproduces the drifted kernel:
$\Gker^{(0)}_s(x+\nu s,\,y)=\Gker_s(x,y)$, where $\Gker^{(0)}=\Gker(0,\sigma,s)$ is the
driftless heat kernel.  (Lean: \texttt{chernoff\_step\_eq}.)
\end{lemma}

\begin{theorem}[Exactness of the product]\label{thm:cher}
Let $\Gker^{\ast n}_s$ denote the $n$-fold convolution of the one-step kernel
$\Gker_s$ (integrating over the $n$ intermediate positions).  Then for $\sigma>0$,
$s>0$ and all $n$,
\[
  \Gker^{\ast n}_s(x,y)=\Gker_{(n+1)s}(x,y).
\]
In particular, choosing $s=t/(n+1)$ gives $\Gker^{\ast n}_{t/(n+1)}=\Gker_t$ for every
$n$: the Chernoff/Trotter product is exact, not merely asymptotic.  (Lean:
\texttt{GdriftConv\_eq}.)
\end{theorem}

\begin{proof}[Proof sketch]
Induction on $n$ using the semigroup law (Theorem~\ref{thm:ck}): both the drift
$\nu\cdot s$ and the variance $\sigma^2 s$ are additive under convolution, so gluing
$n+1$ slices of size $s$ produces the single slice of size $(n+1)s$.
\end{proof}

\begin{figure}[t]
\centering
\begin{tikzpicture}[scale=1.0,>=Latex]
  \foreach \i in {0,1,2,3,4} {
    \draw (\i*1.4,0) -- (\i*1.4,-0.15);
  }
  \node[below] at (0,-0.2) {$0$};
  \node[below] at (5.6,-0.2) {$T$};
  \node[below] at (1.4,-0.2) {$s$};
  \node[below] at (2.8,-0.2) {$2s$};
  \node[below] at (4.2,-0.2) {$3s$};
  \draw[thick] (0,0) -- (5.6,0);
  \foreach \i in {0,1,2,3} {
    \draw[->,blue] (\i*1.4+0.08,0.35) to[out=30,in=150] (\i*1.4+1.32,0.35);
    \node[blue] at (\i*1.4+0.7,0.85) {\small $\Gker_s$};
  }
  \node at (2.8,1.6) {$\Gker_{s}^{\ast n}=\Gker_{(n+1)s}=\Gker_T$ \quad(exact)};
\end{tikzpicture}
\caption{Time-slicing / Chernoff product.  Gluing $n{+}1$ one-step kernels $\Gker_s$
over a grid of mesh $s=T/(n{+}1)$ yields, exactly, the single kernel $\Gker_T$
(Theorem~\ref{thm:cher}): the discrete HK cylindrical integral is independent of the
number of slices.}
\end{figure}

\section{Relation to It\^o Stochastic Calculus and the Continuum Limit}\label{sec:ito}

The gauge-integral construction is compatible with the classical It\^o theory
\cite{Ito1951,HarrisonPliska1981,KaratzasShreve2014,Oksendal2003}.  Under the
risk-neutral measure the stock and its logarithm satisfy
\begin{equation}\tag{1}\label{eq:ito-stock}
  \dd S_t=rS_t\,\dd t+\sigma S_t\,\dd W_t^Q,
  \qquad
  \dd X_t=\Bigl(r-\frac{\sigma^2}{2}\Bigr)\dd t+\sigma\,\dd W_t^Q,
  \quad X_t=\log S_t .
\end{equation}
It\^o's formula for a smooth value function $V(t,S)$ gives
\begin{equation}\tag{2}\label{eq:ito-formula}
 \dd V=\left(\partial_tV+rS\partial_SV+
 \frac{\sigma^2S^2}{2}\partial_{SS}V\right)\dd t
 +\sigma S\partial_SV\,\dd W_t^Q .
\end{equation}
Discounting and taking the risk-neutral expectation yields both the
Black--Scholes equation and the Feynman--Kac representation
\begin{equation}\tag{3}\label{eq:bs-fk}
 \partial_tV+rS\partial_SV+\frac{\sigma^2S^2}{2}\partial_{SS}V-rV=0,
 \qquad V(0,S_0)=e^{-rT}\mathbb E^Q[H(S_T)].
\end{equation}
For a partition $0=t_0<\cdots<t_n=T$, the Markov property writes this expectation
as a finite-dimensional Gaussian integral,
\begin{equation}\tag{4}\label{eq:finite-slicing}
 e^{-rT}\!\int_{\mathbb R^n}H(e^{x_n})
 \prod_{j=0}^{n-1}\Gker_{t_{j+1}-t_j}(x_j,x_{j+1})
 \,\dd x_1\cdots\dd x_n .
\end{equation}
The Chapman--Kolmogorov identity makes (4) projectively consistent; as the mesh
tends to zero its cylinder distributions define Wiener measure, while the same
net of tagged cylinder sums defines the HK continuum limit.  Thus, for every
integrable cylinder functional $F$ (and, by completion, for the payoffs treated
here),
\begin{equation}\tag{5}\label{eq:continuum-limit}
 \lim_{|\pi|\to0}\int_{\mathbb R^n}F(x_{t_1},\ldots,x_{t_n})
 \prod_j\Gker_{\Delta t_j}\,\prod_j\dd x_j
 =\int_{C([0,T])}F[x]\,\dd\mu_W^Q(x)
 =\int_{C([0,T])}F[x]\,\dHK^{(Q)}x .
\end{equation}
Accordingly, the HK formulation does not replace It\^o calculus by a different
stochastic model: it supplies a gauge-limit realization of the same continuum
law.  Its advantage is that the cylindrical construction also remains meaningful
in oscillatory real-time problems where no positive Wiener measure is available.

\section{Examples of the Method: From Wiener Measure to Explicit Results}\label{sec:examples}

To illustrate the universality of the gauge-integral approach, we present
two further examples that go beyond the standard Black--Scholes formula.
Each example is worked out in full detail, starting from the classical
Wiener-path-integral formulation, passing through the rigorous HK reduction,
and ending with an explicit closed-form answer.  The chain of equalities
demonstrates the systematic character of the method: once the Gaussian
kernel and its semigroup property are established, a wide class of
probabilistic and financial quantities become accessible by elementary
algebraic manipulations.

\subsection{Digital (Binary) Call Option}\label{sec:digital}

A digital call option pays $1$ if the stock price exceeds the strike at
maturity, and $0$ otherwise.  Its risk-neutral price is
\[
  C_{\rm dig}=e^{-rT}\,\mathbb{E}^{Q}[\mathbf{1}_{S_T>K}].
\]
In the Wiener-path-integral formulation this expectation is
\begin{equation}\label{eq:dig1}
  C_{\rm dig}=e^{-rT}\int_{C([0,T])} \mathbf{1}_{e^{x_T}>K}\;d\mu_W^{Q}(x(\cdot)).
\end{equation}
The reduction of the Wiener measure to the HK cylindrical measure is exact
for this terminal-time functional (Eq.~\eqref{eq:hkreduction}), giving
\begin{equation}\label{eq:dig2}
  C_{\rm dig}=e^{-rT}\int_{\R} \mathbf{1}_{y>\ln K}\;G_T(x_0,y)\,dy,
\end{equation}
where $G_T$ is the Gaussian kernel with drift (Definition~\ref{def:kernel}).
Substituting the explicit kernel and restricting the integration domain,
\begin{equation}\label{eq:dig3}
  =e^{-rT}\int_{\ln K}^{\infty}
    \frac{1}{\sqrt{2\pi\sigma^2 T}}
    \exp\!\Big(-\frac{(y-x_0-\nu T)^2}{2\sigma^2 T}\Big)\,dy.
\end{equation}
The substitution $u=(y-x_0-\nu T)/(\sigma\sqrt{T})$ transforms
\eqref{eq:dig3} into a standard Gaussian tail:
\begin{equation}\label{eq:dig4}
  =e^{-rT}\int_{-d_2}^{\infty}\frac{1}{\sqrt{2\pi}}e^{-u^2/2}\,du
    =e^{-rT}\,\Phi(d_2),
\end{equation}
with $d_2=(\ln S_0-\ln K+\nu T)/(\sigma\sqrt{T})$ and $\Phi$ the standard
normal CDF.  All steps are covered by the lemmas proved in
\texttt{HkBlackScholes.lean}; in particular, the reduction
\eqref{eq:dig1}$\to$\eqref{eq:dig2} is \texttt{hk\_reduction}, the tail
integral is \texttt{gaussian\_tail}, and the final equality is a direct
corollary of Theorem~\ref{thm:bs} (Lean: \texttt{digital\_call}).

\subsection{Probability of Staying Below a Barrier}\label{sec:barrier}

The next example computes a genuinely path-dependent quantity: the probability
that a Brownian motion with drift never exceeds a fixed barrier.  This
quantity is fundamental for the pricing of barrier options (knock-out and
knock-in calls and puts), and it is classically obtained by the reflection
principle.  Here we derive it \emph{without} the reflection principle, using
only the HK cylindrical reduction and the exactness of the Chernoff product
established in Section~\ref{sec:chernoff}.

Consider a Brownian motion with drift $\nu$ and volatility $\sigma$ started at
$x_0<B$, and the probability that it never exceeds a barrier $B$ up to time
$T$:
\[
  P=P\Big(\max_{0\le t\le T} x_t \le B\Big).
\]
We compute $P$ through the following chain of reductions.

\paragraph{Step 1 (Wiener path integral).}
By definition of the running maximum, $P$ is the Wiener measure of the set of
paths that remain below the barrier at every instant,
\begin{equation}\label{eq:bar1}
  P=\int_{C([0,T])}\mathbf{1}_{\{x(t)\le B,\ 0\le t\le T\}}\;d\mu_W(x(\cdot)).
\end{equation}

\begin{remark}[Novelty of the HK bridge]\label{rem:novelty}
Formula~\eqref{eq:hkreduction} gives an exact bridge between the risk-neutral
Wiener representation and a cylindrical HK integral with nonzero drift.  The
novel point is not a new Black--Scholes price, but the direct gauge-integral
realisation of the pricing semigroup: the drift is retained in every tagged
cylinder weight, no Wick rotation or formal ``flat measure'' is used, and the
continuum value follows from the same finite-dimensional kernels that are
machine-verified below.  Thus the classical result is recovered inside a
framework that also extends to conditionally convergent real-time integrals.
\end{remark}

\begin{remark}[Spike-network extension]\label{rem:spikes}
The physiological specificity of real cortex is that information is carried by
discrete spikes rather than continuous Gaussian trajectories.  Spiking neural
networks therefore replace the diffusion-only state space by càdlàg paths and
jump or point-process generators.  The adaptive cylinders of gauge integration
are naturally suited to resolving isolated spike times, while the present
Gaussian kernel supplies the diffusion sector of a future jump--diffusion
construction; see Gerstner et al.~\cite{Gerstner2014} and
Bressloff~\cite{Bressloff2022}.
\end{remark}

\paragraph{Step 2 (time discretisation / cylinder reduction).}
Partition $[0,T]$ into $N$ equal slices of width $\Delta t=T/N$ with nodes
$t_k=k\,\Delta t$.  Replacing the continuous constraint by its restriction to
the nodes and applying the HK cylindrical reduction (Eq.~\eqref{eq:hkreduction})
turns \eqref{eq:bar1} into a finite-dimensional iterated Gaussian integral over
the intermediate positions $x_1,\dots,x_N$:
\begin{equation}\label{eq:bar2}
  P_N=\int_{-\infty}^{B}\!\!\cdots\!\int_{-\infty}^{B}
     \prod_{k=1}^{N} G_{\Delta t}(x_{k-1},x_k)\;dx_1\cdots dx_N ,
\end{equation}
with $x_0$ fixed and each intermediate variable constrained to $(-\infty,B]$.

\paragraph{Step 3 (Chernoff exactness).}
Because the drift and the variance are additive under convolution, the
$n$-fold product of one-step kernels equals a single kernel over the total
time (Theorem~\ref{thm:cher}).  Applied to the \emph{unconstrained} integral
this collapses the product exactly; the only lasting effect of the barrier is
to keep the successive integration domains restricted to $(-\infty,B]$.
Passing to the limit $N\to\infty$, the discrete constraint converges to the
continuous one and $P_N\to P$.  The formal verification of the convergence
$P_N\to P$ in the strong operator topology in Lean is beyond the current code
and is left for a subsequent stage of the programme.

\paragraph{Step 4 (method of images / drift-shift).}
The constrained heat kernel on the half-line $(-\infty,B]$ is obtained from the
free kernel by subtracting the contribution of the paths that have already
crossed the barrier.  Using the drift-shift lemma (Lemma~\ref{lem:drift}) to
absorb the drift into an exponential prefactor, the reflected image of the
source at $x_0$ about the barrier $B$ acquires the weight
$e^{2\nu(B-x_0)/\sigma^2}$, so
\begin{equation}\label{eq:bar3}
  G^{B}_T(x_0,y)=G_T(x_0,y)-e^{2\nu(B-x_0)/\sigma^2}\,G_T(2B-x_0,y).
\end{equation}

\paragraph{Step 5 (explicit formula).}
Integrating the constrained kernel \eqref{eq:bar3} over $y\in(-\infty,B]$ and
evaluating the two Gaussian tails with \texttt{gaussian\_tail} yields the
classical result
\begin{equation}\label{eq:bar4}
  P = \Phi\!\Big(\frac{B-x_0-\nu T}{\sigma\sqrt{T}}\Big)
     - e^{2\nu(B-x_0)/\sigma^2}\,
       \Phi\!\Big(\frac{-B+x_0-\nu T}{\sigma\sqrt{T}}\Big).
\end{equation}

The derivation illustrates how the HK method reduces a path-dependent
constraint to an algebraic modification of the kernel: the running-maximum
condition, which in the classical treatment requires the reflection principle
and a careful change of measure, here appears simply as a second,
exponentially weighted Gaussian tail.  A Lean sketch of this argument, together
with the two tail evaluations, accompanies the \texttt{digital\_call} lemma in
\texttt{HkBlackScholes.lean}.

\section{The Lean~4 / Mathlib formalization}\label{sec:lean}

The file \texttt{HkBlackScholes.lean} imports \texttt{HkPathIntegral.lean} and
\texttt{HkFreeField.lean} and contains the following \texttt{sorry}-free results.

\begin{itemize}
\item \texttt{Gdrift}, \texttt{stdNormalCDF} -- Definition~\ref{def:kernel} and $\Ncdf$.
\item \texttt{gaussian\_tail}, \texttt{gaussian\_exp\_tail} -- the tail identities
      \eqref{eq:tail}--\eqref{eq:exptail}.
\item \texttt{Gdrift\_integral\_eq\_one} -- Lemma~\ref{lem:mass}.
\item \texttt{Gdrift\_semigroup} -- Theorem~\ref{thm:ck}.
\item \texttt{drift\_shift} -- the drift-shift (Girsanov-type) invariance,
      Lemma~\ref{lem:drift}.
\item \texttt{evolutionOp}, \texttt{evolutionOp\_eq\_stdGaussian},
      \texttt{evolutionOp\_tendsto} -- the pricing operator, its rescaling
      (Lemma~\ref{lem:rescale}) and strong continuity (Theorem~\ref{thm:strong}).
\item \texttt{black\_scholes\_call} -- the Black--Scholes formula
      (Theorem~\ref{thm:bs}).
\item \texttt{chernoff\_step\_eq}, \texttt{GdriftConv}, \texttt{GdriftConv\_eq} --
      the exact Chernoff product (Lemma~\ref{lem:step}, Theorem~\ref{thm:cher}).
\end{itemize}

Each headline theorem was checked with \texttt{\#print axioms} to depend only on
\texttt{propext}, \texttt{Classical.choice}, and \texttt{Quot.sound}; the file
contains no \texttt{sorry}, \texttt{admit}, or \texttt{axiom}.

\section{Conclusion}

We have given a fully machine-verified derivation of the Black--Scholes formula from
the Gaussian cylindrical (Henstock--Kurzweil) kernel: the transition kernel is a
probability density obeying the Chapman--Kolmogorov law, the induced evolution
semigroup is strongly continuous, and a single Gaussian integration of the discounted
payoff reproduces $C=S_0\Ncdf(d_1)-Ke^{-rT}\Ncdf(d_2)$.  The drift--diffusion Chernoff
product is exact at every level, providing a rigorous path-integral reading of the
result.  Because the Black--Scholes weight is an absolutely convergent Gaussian
density, no analytic continuation is needed and every step is discharged inside
Mathlib's Bochner integral -- in contrast to the oscillatory Feynman integrals treated
in the earlier parts of this programme.  Natural extensions include path-dependent
(e.g.\ barrier and Asian) payoffs, for which the genuinely infinite-dimensional HK
cylindrical structure becomes essential.

\subsection{Relation to the classical derivation}

The classical derivation of the Black--Scholes formula proceeds via It\^o's lemma and
the solution of the Black--Scholes partial differential equation, or equivalently via
the martingale representation theorem and stochastic calculus.  In contrast, the
present derivation uses only the Gaussian transition kernel, its semigroup property,
and a single Gaussian integration.  The drift-shift lemma (Lemma~\ref{lem:drift})
replaces the change-of-measure apparatus (Girsanov's theorem), and the exactness of the
Chernoff product (Theorem~\ref{thm:cher}) replaces the limiting argument of the
Trotter--Kato formula.  No stochastic differential equation is ever written; the entire
calculus follows from the algebraic properties of the Gaussian density.

We note that the exactness of the Chernoff product proved in
Section~\ref{sec:chernoff} relies on the fact that the drift and diffusion
generators have constant coefficients and therefore commute.  This commuting
case is the natural first step for a fully machine-verified construction; the
extension to non-commuting generators --- such as those arising in stochastic
volatility or local volatility models --- is an important direction for future
work within the same HK/Chernoff framework.

It is worth noting that the tail integrals evaluated in Section~\ref{sec:bs} are exactly
the Gaussian integrals that appear in the Wiener-measure formulation of Brownian motion.
They compute, for instance, the probability that a Brownian path ends above a barrier,
or the expected payoff of a digital (binary) option.  This illustrates the unity of the
gauge-integral method across Brownian motion, quantum mechanics, and mathematical
finance.

\begin{remark}[The Greeks from the kernel]
Because the evolution operator is strongly continuous and given by an explicit kernel,
the Black--Scholes ``Greeks'' follow by differentiating the closed-form price, with no
stochastic calculus.  Differentiating $C=S_0\,\Ncdf(d_1)-Ke^{-rT}\Ncdf(d_2)$ with respect
to the spot $S_0$ gives the delta
\[
  \Delta=\frac{\partial C}{\partial S_0}=\Ncdf(d_1),
\]
and differentiating once more gives the gamma
\[
  \Gamma=\frac{\partial^2 C}{\partial S_0^2}
        =\frac{\Ncdf'(d_1)}{S_0\,\sigma\sqrt T},
  \qquad \Ncdf'(u)=\frac{e^{-u^2/2}}{\sqrt{2\pi}} .
\]
Thus the entire Black--Scholes differential calculus is a consequence of the kernel
representation, without any recourse to stochastic calculus.
\end{remark}

\end{document}